\documentclass[pdftex,twocolumn,epjc3]{svjour3}

\RequirePackage{amsmath,amssymb,amsfonts}
\RequirePackage[T1]{fontenc}
\RequirePackage{graphicx}
\RequirePackage{mathptmx}
\RequirePackage{flushend}
\RequirePackage[numbers,sort&compress]{natbib}
\RequirePackage[colorlinks,citecolor=blue,
urlcolor=blue,linkcolor=blue]{hyperref}
\smartqed

\journalname{Eur. Phys. J. C}

\begin{document}

\title{Bound orbits near scalar field naked singularities}

\author{I.\,M. Potashov\thanksref{addr1,e1} \and
Ju.\,V. Tchemarina\thanksref{addr1,e2}      \and
A.\,N. Tsirulev\thanksref{addr1,e3}       }

\thankstext{e1}{e-mail: potashov.im@tversu.ru}
\thankstext{e2}{e-mail: chemarina.yv@tversu.ru}
\thankstext{e3}{e-mail: tsirulev.an@tversu.ru}

\institute{Faculty of Mathematics, Tver State University, Sadovyi 35, Tver, Russia \label{addr1}
}

\date{Received: date / Accepted: date}

\maketitle

\begin{abstract}
We study bound orbits near the centres of static, spherically symmetric, asymptotically flat configurations of a self-gravitating scalar field minimally coupled to gravity. In our approach, a nonlinear scalar field is considered as an idealized model of dark matter, and the main examples that we have in mind are the centres of galaxies. We consider both scalar field black holes (SFBHs) and scalar field naked singularities (SFNSs). It turns out that the shape and parameters of a bound orbit depend crucially on the type of configuration. The lapse metric function of a SFNS and, consequently, the effective potential of a massive test particle with zero angular momentum have a global minimum. A SFNS has a static degenerated orbit on which a test particle, having zero angular momentum and the minimum of its energy, remains at rest at all times. This implies that there exists a spherical shell consisting of cold gas or dust, which for a distant observer can look like the shadow of a black hole. We also study the shape of noncircular bound orbits close to the centres of SFNSs and show that their angles of precession are negative.
\end{abstract}

%At the present time, there are known a number of scalar hairy solutions of the EKG equations, however, most of them are unstable and, second, .

\section{Introduction}

At present, we still do not exactly know the geometry of spacetime in the neighbourhood of the centre of normal galaxies. The best observational results have been obtained for the centres of our Galaxy and of M87 \cite{Meyer2012,Akiyama2015,Fish2016, Gillessen2017,Goddi2017,Hees2017}. However, the available data are so far insufficient to identify these strongly gravitating objects and even to definitely distinguish between black holes, naked singularities, boson stars, and wormholes~\cite
{Joshi2014,Macedo2013,Li2014,Grould2017a, Grould2017b,Abuter2018}. The EHT collaboration has recently observed~\cite{Akiyama2018} the shadow, circled by a bright ring, in the centre of M87: this can be interpreted as the existence of closed photon orbits. However, it is shown in~\cite{Shaikh2018} using a simple model that a naked singularity can have both a shadow and a photon sphere. In fact, observations of the orbits of stars near the centres have a key role in dealing with this question, but there are some obvious problems with the reasonable geometrical interpretation of such observations. First, one should not think of the central objects in galaxies as being in vacuum, because dark matter is mainly concentrated around them. Another problem is that the nature of dark matter and its distribution near galactic centres remain unknown at present. This means that a meaningful interpretation of the observations should be based on an appropriate mathematical model of the central regions.

In this paper, we model dark matter by a nonlinear scalar field which is assumed to be minimally coupled to gravity. This model is an interesting alternative to the cold dark matter phenomenology (see e.g. \cite{Lee1996, Matos2004, Robles2012, Benisty2017, Bernal2019} and references therein). In the centres of galaxies, the distribution of dark matter seems to be spherically symmetric. Our aim is to study bound orbits of test particles near the centres of static, spherically symmetric configurations of a self-gravitating nonlinear scalar field; in our terminology, a bound orbit is a bounded \textit{complete} geodesic in an asymptotically flat spacetime. The motivation for the choice of this special and idealized model is to treat the problem in a fully analytical and self-consistent manner. For the orbital motion of a test particle near a compact hairy object, this model allows us to find out some important features, which can possibly be observed in real astrophysical systems, and which appear to be hidden in a purely numerical analysis.

When we speak about SFNSs and SFBHs, we are faced with some unsolved problems, such as the cosmic censorship conjecture, instability of configurations, violation of the nonnegativity of self-interaction potentials, and so on. However, a number of these problems, including the singularity problem, disappear if we assume that there is some additional matter (perhaps of an unknown nature, beyond the Standard Model) in the very centre. Note also that we regard, in what follows, a scalar field as a phenomenological construction rather than a fundamental natural field;  on the other hand, the Standard Model of particle physics predicts the existence of such a field, and it is of great importance for the modern cosmology. In both the cases, we do not know (in advance, without appealing to astrophysical observations) the self-interaction potential of the scalar field. Therefore, in order to consider the problem in a sufficiently general approach, we use the so-called inverse problem method for self-gravitating spherically symmetric scalar fields; the method was proposed in
\cite{Lechtenfeld1995,Lechtenfeld1998, BronnikovShikin2002}
and later was explored in
\cite{Tchemarina2009,Azreg2010,Cadoni2011, Solovyev2012,Cadoni2018} and applied, for example, in \cite{BronnikovChernakova2007, Nikonov2008,Franzin2018}.

The paper is organised as follows. In Sect.~\ref{Sec2} we describe the necessary mathematical background for static, spherically symmetric scalar field configurations, and obtain some general results both for SFBHs and for SFNSs. One of the most important issues is how to distinguish a naked singularity from a black hole using available observational data. On the other hand, our primary purpose is the study of geodesic motion in these asymptotically flat spacetimes, and in Sect.~\ref{Sec3} we discuss general features of bound orbits of free massive particles, such as the shape of the orbits and the precession of their pericentres. In Sect.~\ref{Sec4} we consider a simple, fully analytical, one-parameter family of SFNSs and study the shapes of bound orbits in these spacetimes in comparison with the orbits in the corresponding Schwarzschild black hole spacetimes.

In this paper, we use the geometrical system of units with $G=c=1$ and adopt the metric signature $(+\,-\,-\,-)$.

\section{Comparing SFBHs and SFNSs}
\label{Sec2}

The action with the minimal coupling between curvature and a real scalar field $\phi$ has the form
\begin{equation}\label{action}
    \Sigma=\frac{1}{8\pi}\int\left(-\frac{1}{2}R + \,\langle d\phi,d\phi\rangle-2V(\phi)\right) \sqrt[]{|g|}\,d^{\,4}x\,,
\end{equation}
where $R$ is the scalar curvature, $V(\phi)$ is a self-interaction potential, and the angle brackets denote the scalar product with respect to the spacetime metric. For our purposes, it is convenient to write the metric of a spherically symmetric spacetime in the Schwarzschild-like coordinates as
\begin{equation}\label{metric}
    ds^2=A dt^2- \frac{\,dr^2}{f}- r^{2}(d\theta^2+\sin^{2}\!\theta\, d\varphi^2),
\end{equation}
where the metric functions $A$ and $f$ depend only on the radial coordinate $r$. Writing $$A(r)=f(r)\mathrm{e}^{2\Phi(r)},$$
we obtain the Einstein-Klein-Gordon equations in the form
\begin{equation}\label{EKG1}
-\frac{f'}{r}-\frac{f-1}{r^2}={\phi'}^2 f + 2V\,,
\end{equation}
\begin{equation}\label{EKG2}
\frac{f}{r}\left(\!2\Phi'+\frac{f'}{f}\right)+\frac{f-1}{r^2}=
{\phi'}^2 f - 2V\,,
\end{equation}
\begin{equation}\label{EKG3}
-f\phi''-\frac{\phi'}{2}f'-\phi' f\left(\!\Phi' +
\frac{1}{2}\frac{f'}{f}+\frac{2}{r}\right)+ \frac{dV}{d\phi}=0\,,
\end{equation}
where a prime denotes differentiation with respect to $r$. Now these equations can be reduced to two independent ones by summing equations (\ref{EKG1}) and (\ref{EKG2})\, (with the result $\Phi'=r{\phi'}^2$), and then eliminating $\Phi'$ from (\ref{EKG3}).

The functions $A$ and $f$, which completely determine geodesic motion in the spacetime, should be the result of solving equations (\ref{EKG1})\,--\,(\ref{EKG3}). However, we have no a priori knowledge of the form of $V(\phi)$ and have to study bound orbits for all physically admissible self-interaction potentials at once. Our approach is based on the following result~\cite{Solovyev2012}:

\textit{A general static, spherically symmetric, asymptotically flat solution of equations (\ref{EKG1})\,--\,(\ref{EKG3}) with an arbitrary self-interaction potential is given by the quadratures}
\begin{equation}\label{F-xi}
\Phi(r)=-\!\int_{r}^{\,\infty} {\phi'}^{2}rdr\,,\quad \xi(r)=r+\int_{r}^{\,\infty}
\left(1-\mathrm{e}^{\Phi}\right)\!dr\,,
\end{equation}
\begin{equation}\label{A-f}
A(r)=2r^{2}\!\int_{r}^{\,\infty} \frac{\,\xi-3m}{\,r^4}\,\mathrm{e}^{\Phi}dr\,, \quad f(r)=\mathrm{e}^{-2\Phi}A\,,
\end{equation}
\begin{equation}\label{V}
\widetilde{V}(r)=\frac{1}{2r^2}\!\left(1-3f+ r^2{\phi'}^{2}\!f+ 2\,\mathrm{e}^{-\Phi}\,\frac{\,\xi-3m}{r}\right),
\end{equation}
\textit{where the parameter $m$ is the Schwarzschild mass}.

\noindent
Each strictly monotonic\footnote{Note, however, that each solution of the problem under consideration satisfies the quadratures (\ref{F-xi})\,--\,(\ref{V}) regardless of the monotonicity of the field function.} function $\phi(r)$ of class $C^2\big([0,\infty)\big)$ with the asymptotic behaviour
\begin{equation}\label{phi}
\phi= O\!\left(r^{-1/2-\textstyle{\alpha}}\right)\!,\;\;\, r\rightarrow\infty\qquad (\alpha>0)
\end{equation}
determines a one-parameter family of solutions by these quadratures: one sequentially finds the functions $\mathrm{e}^{\Phi}$, $\xi$, $A$, $f$, and $\widetilde{V}(r)$, and then find the potential $V(\phi)=\widetilde{V}(r(\phi))$. Note also that we could include the cosmological constant in the potential as the additional term $\Lambda/2$, but its contribution to the geometry of the central region would be negligible. The absence of the cosmological constant simply means that $V(\phi(\infty))=0$.

In spherically symmetric spacetimes, the Kretchmann invariant, $K= R_{\alpha\beta\gamma\delta} R^{\alpha\beta\gamma\delta}/4$, equals the sum of the squared curvature components. In the orthonormal basis associated with the metric~(\ref{metric}), algebraically independent components of the curvature can be reduced with the aid of (\ref{F-xi}) and (\ref{A-f}) to the form
\begin{eqnarray}
R_{0101}&=& {\phi'}^{2}f-\frac{f-1}{r^2}, \quad R_{2323}=\frac{f-1}{r^2},
\label{R1}\\
R_{0202}&=&R_{0303}= -\frac{f}{r^2}+ \mathrm{e}^{-\Phi}\,\frac{\xi-3m}{r^3},
\label{R2}\\
R_{1212}&=&R_{1313}=\frac{f}{r^2}- {\phi'}^{2}f- \mathrm{e}^{-\Phi}\,\frac{\xi-3m}{r^3},
\label{R3}
\end{eqnarray}
One can see that with the exception of some special fine-tuned cases, $K$ and $R$ diverge at $r=0$ for all solutions.  In the generally accepted manner, we call a solution \textit{a naked singularity} (\textit{a black hole}) if $K$ diverges at $r=0$ and $f>0$ for all $r>0$ \big(respectively, $f=0$ at some radius $r_{\!\vphantom{\hat{A}}h}>0$ and $f>0$ for all $r>r_{\!\vphantom{\hat{A}}h}$\big).

For a given nonzero scalar field $\phi(r)$, it follows directly from~(\ref{F-xi}) that $\xi'=\mathrm{e}^{\Phi}>0$ for all $r>0$ and $\xi(0)>0$, so that the metric function $A$, given by the quadrature~(\ref{A-f}), passes through zero and becomes negative as $\,r\rightarrow0\,$ if and only if $\,\,3m>\xi(0)$.
In other words, the corresponding configuration of mass $m$ will be a naked singularity or a black hole if
\begin{equation}\label{3m<xi0}
0<3m<\xi(0)\quad \mbox{(naked singularities)}
\end{equation}
or
\begin{equation}\label{3m>xi0}
3m>\xi(0)\quad \mbox{(black holes)}
\end{equation}
respectively. In what follows we deal only with 'generic' configurations and do not consider the special (fine-tuned) case $3m=\xi(0)$; the latter leads to a naked singularity or a regular solution.

It follows from~(\ref{F-xi}) and (\ref{phi}) that
\begin{eqnarray}\label{est1}
\mathrm{e}^{\Phi}&=&1+o(1/r),\quad \xi\!=r+o(1),\;\;\; r\rightarrow\infty,\vphantom{\int}\\
\label{est2}
\mathrm{e}^{\Phi}&=&\mathrm{e}^{\Phi(0)}+o(r),\;\; \xi\!=\xi(0)+\xi'(0)r+o(r),\;\; r\rightarrow0,
\end{eqnarray}
where $\xi(0)>0$ and $\xi'(0)=\mathrm{e}^{\Phi(0)}>0$ if the scalar field is not identically zero. Using the quadrature~(\ref{A-f}) and the asymptotic estimates (\ref{est1}) and (\ref{est2}), one can directly obtain the asymptotic behaviours of the metric function $A(r)$ for both SFNSs and SFBHs; they are
\begin{eqnarray}\label{asymp1}
A(r)&=&1-\frac{2m}{r}\vphantom{\int\limits_a}+ o(1/r),\;\; r\rightarrow\infty,\\
\label{asymp2}
A(r)&=&\frac{2}{3}\:\! \frac{\xi(0)-3m}{r}\,\mathrm{e}^{\Phi(0)}\,+\, O(1),\;\; r\rightarrow0.
\end{eqnarray}
The condition $m>0$ distinguishes SFNSs from the vacuum (that is, Schwarzschild) naked singularities, which exist only for negative values of mass. Below we will also need a few inequalities and expressions for $\xi$ and $A$. It follows directly from~(\ref{F-xi}) and~(\ref{A-f}) that for all $r>0$
\begin{eqnarray}\label{cond1}
\xi\,\,&\!\!\!\geq&r\,,\;\;\; 0<\xi'\leq1,\;\;\;  0<\mathrm{e}^{\Phi}\leq1,\;\;\; \xi''= \Big(\!\mathrm{e}^{\Phi}\Big)^{\!\prime}\geq0\,,\\
\label{A'}
A'&=& \frac{2}{r}\,A- 2\,\frac{\xi-3m}{r^2}\,\mathrm{e}^{\Phi}\,,\\
\nonumber
A''&=& \frac{2}{r}\,A'- \frac{2}{r^2}\,A + \frac{2}{r^2}\,\frac{\xi-3m}{r}\,\mathrm{e}^{\Phi}- \frac{2}{r} \left(\frac{\xi-3m}{r}\,\mathrm{e}^{\Phi}\right)'\\
\label{A"}
&=& \frac{A'}{r}- \frac{2}{r} \left(\frac{\xi-3m}{r}\,\mathrm{e}^{\Phi}\right)'.
\end{eqnarray}

The two following propositions give us an additional characteristic feature, besides the existence or non-existence of an event horizon and an innermost stable circular orbit, that distinguishes SFNSs from SFBHs.

\begin{proposition}\label{proposition1}
In a SFBH spacetime defined by the quadratures (\ref{F-xi})\,--\,(\ref{V}) and the conditions~(\ref{phi}) and (\ref{3m>xi0}), $A(r)$ is a strictly increasing function outside the event horizon.
\end{proposition}

\begin{proof}
In the region where $\xi-3m\,\leq\,0$, the monotonicity (outside the horizon) follows directly from (\ref{A'}), so we need to consider only the region where $\xi-3m>0$.

An integration by parts in (\ref{F-xi}) yields
\begin{eqnarray}
\xi(r)=r+\int_{r}^{\,\infty} \left(1\!-\mathrm{e}^{\Phi}\right)\!dr = r\mathrm{e}^{\Phi}+ \int_{r}^{\,\infty} \Big(\!\mathrm{e}^{\Phi}\Big)^{\!\prime}rdr
\nonumber\\
=r\mathrm{e}^{\Phi}+ \int_{0}^{\,\infty} \Big(\!\mathrm{e}^{\Phi}\Big)^{\!\prime}rdr- \int_{0}^{\,r} \Big(\!\mathrm{e}^{\Phi}\Big)^{\!\prime}rdr
\nonumber\\
=r\mathrm{e}^{\Phi}+ \xi(0)- \int_{0}^{\,r} \Big(\!\mathrm{e}^{\Phi}\Big)^{\!\prime}rdr\,,
\end{eqnarray}
\begin{eqnarray}
r\mathrm{e}^{\Phi}=\xi(r)-\xi(0)+\int_{0}^{\,r} \Big(\!\mathrm{e}^{\Phi}\Big)^{\!\prime}rdr\,. \nonumber
\end{eqnarray}
Thus, for SFBHs ($\xi(0)<3m$), we have the inequality
\begin{equation}\label{ineq1}
\mathrm{e}^{\Phi}\,>\,\frac{\xi-3m}{r} \quad\text{for all}\;\; r>0\,.
\end{equation}

Integrating by parts in (\ref{A-f}) and then applying the identity $(\xi-3m)'=\mathrm{e}^{\Phi}$, we obtain
\begin{eqnarray}
A&=&2r^{2}\!\int_{r}^{\,\infty} \frac{\,\xi-3m}{\,r^4}\,\mathrm{e}^{\Phi}dr
\vphantom{\int\limits_{r}^{\infty}}
\nonumber\\
&=&\frac{2}{3} \frac{\,\xi-3m}{\,r}\mathrm{e}^{\Phi}+ \frac{2r^{2}}{3}\!\int_{r}^{\,\infty}\! \frac{\big[(\xi\!-\!3m) \mathrm{e}^{\Phi}\big]^{\!\prime}} {\,r^3}\,dr \vphantom{\int\limits_{r}^{\infty}}
\nonumber\\
&=&\frac{2}{3} \frac{\,\xi-3m}{r}\,\mathrm{e}^{\Phi}+\, \frac{\mathrm{e}^{2\Phi}}{3}\,+\, \frac{\,r^{2}}{3}\!\int_{r}^{\,\infty}
\Big(\!\mathrm{e}^{2\Phi}\Big)^{\!\prime}\frac{\,dr} {\,r^2}
\vphantom{\int\limits_{r}^{\infty}}
\nonumber\\
&&\qquad\qquad\quad\! \;+\,\frac{\,2r^{2}}{3}\!\int_{r}^{\,\infty} \frac{\,\xi-3m}{r^3} \Big(\!\mathrm{e}^{\Phi}\Big)^{\!\prime}dr\,.
\label{Byparts}
\end{eqnarray}
Taking into account~(\ref{cond1}) and~(\ref{ineq1}), we have
\begin{equation}\label{ineq1}
A(r)\,>\,\frac{\xi-3m}{r}\mathrm{e}^{\Phi} \quad\text{for all}\;\; r>0\,.
\nonumber
\end{equation}

As it can be seen from (\ref{A'}), this inequality implies that $A'>0$ in the region where $\xi-3m>0$. \qed
\end{proof}

For a SFNS, the existence of a minimum follows directly from~(\ref{3m<xi0}), (\ref{asymp1}), and~(\ref{asymp2}). Suppose that there are other extrema. Let $r_0$ be the local minimum point farthest from the origin $r=0$, and $r_m$ be the local maximum point nearest to $r_0$, that is, there are no local extrema in the interval $(r_m,r_0)$. Then $A'(r_0)=A'(r_m)=0$, $A''(r_0)>0$, and $A''(r_m)<0$, so that~(\ref{A"}) gives
\begin{equation}\label{cond2}
\left(\frac{\xi-3m}{r}\, \mathrm{e}^{\Phi}\right)'_{r=r_0}< 0\,,\;\;\; \left(\frac{\xi-3m}{r}\, \mathrm{e}^{\Phi}\right)'_{r=r_m}> 0\,.
\nonumber
\end{equation}
Consequently, the minimum at $r=r_0$ is unique if
\begin{equation}\label{cond3}
\left(\frac{\xi-3m}{r}\, \mathrm{e}^{\Phi}\right)'\leq 0\,,\quad 0<r<r_0\,.
\end{equation}
It is important to note that this condition is quite natural and means simply that the configurations of a scalar field is sufficiently compact and the function $\mathrm{e}^{\Phi}$ varies sufficiently smoothly (has no steep jumps). A more detailed characterization in terms of the field function $\phi$ turns out to be still more complicated, but numerous example show that the condition~(\ref{cond3}) holds for physically reasonable configurations. Thus we have stated the following proposition.

\begin{proposition}\label{proposition2}
In a SFNS spacetime defined by the quadratures (\ref{F-xi})\,--\,(\ref{V}) and the conditions (\ref{phi}) and~(\ref{3m<xi0}), the metric\linebreak function $A(r)$ has a global minimum in the interval $(0,\infty)$. If the condition~(\ref{cond3}) holds, this minimum is unique.
\end{proposition}

For simplicity, we will assume below that the metric function $A(r)$ \textit{has exactly one minimum} at $r=r_0$.

It is shown in the next section that these propositions give us a key distinguishing feature of scalar field configurations; in particular, it determines different behaviours of bound orbits around SFNSs on the one hand and around SFBHs on the other.

\section{Bound orbits}
\label{Sec3}

In any static, spherically symmetric spacetime a massive test particle has three integrals of motion. For the metric of the form~(\ref{metric}) they are
\begin{equation}\label{int-motion}
\frac{dt}{ds}=\frac{E}{A}\,,\quad
\frac{d\varphi}{ds}=\frac{J}{\,r^2}\,,\quad
\left(\frac{dr}{ds}\right)^{\!\!2}=  \mathrm{e}^{-2\Phi}\!\! \left(E^2- {V\vphantom{\underline{A}}}_{\!\!eff}\right), \end{equation}
\begin{equation}\label{Veff}
{V\vphantom{\underline{A}}}_{\!\!eff}= A\!\left(1+\frac{\,J^2}{\,r^2}\right),
\end{equation}
where $V\vphantom{\underline{A}}_{\!\!eff}$ is the effective potential, $E$ is the specific energy, and $J$ is the specific angular momentum of a massive test particle.

We note that $V_{eff}\rightarrow1$ as $r\rightarrow\infty$ for any value of $J$ and for any asymptotically flat spacetime, but the results obtained in Sect.~\ref{Sec2} show radically different behaviour of the effective potentials in the interior regions of SFBHs and SFNSs, which implies in turn a crucial distinction between geodesic structures of these spacetimes. In a SFBH spacetime as well as in a Schwarzschild spacetime, the effective potential of a test particle vanishes at the horizon and has, for sufficiently large $J$, at least one minimum and one maximum outside the horizon. The radius $r_{\!\scriptscriptstyle h}$ of the event horizon of a SFBH is always less than that of the vacuum black hole with the same mass (see Fig~\ref{fig1}). It follows directly from~(\ref{asymp2}) that $r_{\!\scriptscriptstyle h}\rightarrow0$ as $m\rightarrow\xi(0)/3+0$  and can be arbitrary close to zero. Moreover, the various numerical simulations with SFBH solutions allow us to conclude that the radius of the corresponding innermost stable circular orbit, which is an important observational characteristic for black holes, is of order $3r_{\!\scriptscriptstyle h}$  (analogously to the vacuum case). On the contrary, a SFNS has no innermost stable circular orbit but has a unique degenerated \textit{static orbit}, which has $J=0$ and is located at $r=r_0$, where $r_0$ is the unique solution of the equation $A'(r)=0$. From the point of view of a distant observer resting relative to the centre, a test particle remains at rest in the static orbit all time. Particles in such a static orbit, together with particles having zero or small specific angular momentum and specific energy $E^2\gtrapprox{}A(r_0)$, can make up a spherical shell consisting of cold gas or fluid. For a distant observer, this shell would look like a shadow similar to that of a black hole. The existence of static degenerated orbits in other spacetimes is considered in~\cite{Quevedo2011, Pugliese2012,  Vieira2014, Collodel2018}. We note also that in SFNS spacetimes, as opposed to black holes, a freely moving massive particle with $E\geq1$ will inevitably escape to infinity, while such a particle with $E<1$ will move on a bound geodesic; thus the value $E=1$ separates bound and unbound orbits.

Another important observational aspect of geodesic motion in the central regions of SFNSs is that orbits with sufficiently large $J$ can lie outside this shell and, at the same time, can have parameters which differ strongly from those for the corresponding orbits in vacuum and SFBHs spacetimes. For a given orbit in a SFNS spacetime, we are primarily interested in the angle of precession $\Delta\varphi$ of the orbit; the latter can be expressed as
\begin{equation}\label{precess}
\varphi_{\!{}_{\scriptscriptstyle osc}}=2J\!\int \limits_{\,r_{\!\scriptscriptstyle min}}^{\:r_{\!\scriptscriptstyle max}}\! \frac{\,\mathrm{e}^{\Phi}} {\,r^2\sqrt{E^2-V_{\!\scriptscriptstyle eff}\,}}\,dr\,,\quad \Delta\varphi\,=\,\varphi_{\!{}_{\scriptscriptstyle osc}}-2\pi\,,
\end{equation}
where $r_{\!\scriptscriptstyle min}$ and $r_{\!\scriptscriptstyle max}$ are, respectively, the pericentre and apocentre radii. In other words, they are solutions of the equation $E^2-V\vphantom{\underline{A}}_{\!\!eff}=0$ such that $r_{\!\scriptscriptstyle min}< r_{\!\scriptscriptstyle J}< r_{\!\scriptscriptstyle max}$, where $r_{\!\scriptscriptstyle J}$ is a (global or local) minimum of $V\vphantom{\underline{A}}_{\!\!eff}(r,J)$, and there are no other solutions in the interval $(r_{\!\scriptscriptstyle min},r_{\!\scriptscriptstyle max})$. Thus, a bound orbit of the general type oscillates near a stable circular orbit (an oscillation is the motion from pericentre to apocentre and back) and $\varphi_{\!{}_{\scriptscriptstyle osc}}$ is the angle between any two successive pericentre points of the orbit. The relativistic precession of pericentres of bound orbits is considered in~\cite{Meyer2012, Goddi2017,  Grould2017b,  Vieira2014, Collodel2018, Zakharov2012, Dokuchaev2015, Zakharov2018} both from a purely theoretical point of view and in the context of observations of S-stars in the Galactic Centre.

If a massive test particle moves radially in some non-static degenerate orbit, which will be the case if $J=0$ and $A(r_0)<E^2<1$, then $\Delta\varphi=-2\pi$. For any small $J\gtrapprox{}0$, this orbit becomes nondegenerate and can have an arbitrarily large number of oscillations per revolution. With further increasing $J$ (and with fixed $E$), the number and the amplitude of oscillations monotonically decrease, while the angle of precession increases and reaches its maximum value. We always have $\Delta\varphi<0$ for orbits located, even if only in part, in the central region.

If an orbit is located in the asymptotic region where the influence of a scalar field on the spacetime geometry vanishes, then the angle of precession is always positive, as well as in the Schwarzschild spacetime. There is also the intermediate spacetime region where the scalar field is relatively weak and the precession is absent for some special values of $J$ and $E$. More exactly, for a given SFNS spacetime, the domain (in the $(J,E)$ plane) of the existence of noncircular bound orbits will be separated into two open parts, with $\Delta\varphi<0$ and $\Delta\varphi>0$, by the curve defined by the equation $\Delta\varphi(E,J)=0$.

At the end of this section, it is important to note that we have a powerful degree of freedom in the choice of the self-interaction potential (in the wide class of physically admissible potentials) of a scalar field or, equivalently, in the choice of the distribution of the scalar field itself. To see a degree of universality of the construction under consideration, one can take into account the identity $2(\xi-3m)\mathrm{e}^{\Phi}=\left[(\xi-3m)^2\right]'$ and rewrite the expression~(\ref{A'}) as an equation for $\xi$ (and thus for $\phi$). The resulting equation,
\begin{equation}\label{}
\left[(\xi-3m)^2\right]'=2rA-r^2A'
\nonumber
\end{equation}
with the asymptotic condition $(\xi-3m)^2\,\rightarrow{}\,r^2-\,6mr$\, as $r\rightarrow\infty$,\, has a unique solution in the region where $2A>rA'$. In this region, using a suitable scalar field distribution, we can simulate some spherically symmetric geometry with a given metric function $A(r)$, and therefore the corresponding effective potential of a test particle for arbitrary specific angular momentum and energy. The condition $2A>rA'$ holds for naked singularities with any physically reasonable gravitating matter and it holds for black holes outside the photon sphere. The effective potential alone does not uniquely determine all the parameters of a bound orbit but only the pericentre and apocentre radii; in addition, we can approximately reconstruct the shape of the orbit by simultaneous varying the angular momentum and energy. In this paper we are mainly interested in just the shape of orbits because the angle of precession cumulates from revolution to revolution. However, there are other important observational parameters, namely, the period of an oscillation, the orbital (tangential) velocity $v_o$ in pericentre, and the radial velocity $v_r$ of a test particle, which can be expressed from~(\ref{int-motion}) in the form
\begin{equation}\label{}
T_{\scriptscriptstyle osc}=2E\!\int \limits_{\,r_{\!\scriptscriptstyle min}}^{\:r_{\!\scriptscriptstyle max}}\! \frac{\,\mathrm{e}^{\Phi}} {\,A\sqrt{E^2-V_{\!\scriptscriptstyle eff}\,}}\,dr\,,
\nonumber
\end{equation}
\begin{equation}\label{}
v_o=\frac{J}{E}\frac{\,A(r_{min})}{\,r}, \qquad
v_r=\frac{dr}{dt}= \frac{1}{E}\! \sqrt{fA\left(E^2- {V\vphantom{\underline{A}}}_{\!\!eff}\right)\,}\,.
\nonumber
\end{equation}
Given an orbit around a SFNS and the orbit with the same $J$, $r_{min}$, and $r_{max}$ around the Schwarzschild black hole of the same mass as the SFNS, it can be shown from~(\ref{A-f}) and~(\ref{cond1}) that $T_{\scriptscriptstyle osc}^{\scriptscriptstyle SFNS}< T_{\scriptscriptstyle osc}^{\scriptscriptstyle Sch}$ and $v_{o}^{\scriptscriptstyle SFNS}> v_{o}^{\scriptscriptstyle Sch}$.

The orbital velocity in pericentre is also determined only by the metric function $A$, while the period of an oscillation and the radial velocity can be determined only if another metric function in~(\ref{A-f}), $f(r)$, is known; however, the latter cannot in general be exactly reconstructed in this way. One has $\mathrm{e}^{\Phi}\leq1$ for all $r>0$, in accordance with the expression~(\ref{F-xi}); the latter is equivalent to the condition $f(r)\geq{}A(r)$, which holds for any physically reasonable configuration except for wormholes. All that has been said allows us to conclude that scalar fields, even if they do not exist in nature, are useful in modelling spherically symmetric, self-gravitating, compact hairy objects, and thus are useful in interpreting the astronomical observations of bound orbits.

\section{ An analytic example}
\label{Sec4}

In a purely analytical treatment of the quadratures~(\ref{A-f})\,--\,(\ref{V}), it is more convenient to start with some specially chosen function~$\xi(r)$ \big(or $\mathrm{e}^{\Phi(r)}$\big) instead of the field function $\phi(r)$, because each of the functions $\mathrm{e}^{\Phi(r)}$, $\xi(r)$, and $\phi(r)$ uniquely determines another two functions.

For the sake of brevity, we will explore a simple, fully analytic, one-parameter family of solutions with
\begin{equation}\label{xi}
\xi= \sqrt{r^2+2ar+5a^2}-a,\quad \mathrm{e}^{\Phi}= \frac{r+a}{\sqrt{r^2+2ar+5a^2}}.
\end{equation}
By direct integration in (\ref{A-f}), we obtain
\begin{multline}\label{Aex}
A=1+\frac{2a}{3r}\,-\,2\,\frac{a+3m}{15a} \vphantom{\int\limits_{r}^{\infty}}\\
\times\left\{\!\!\frac{\sqrt{r^2+2ar+5a^2}}{r}\! \left(\!\!1+\frac{r}{a}-\frac{r^2}{a^2}\!\right)+ \frac{r^2}{a^2}\!\right\},
\end{multline}
where $a$ is the parameter of 'intensity' of the scalar field.

In studying bound orbits, we are interested only in the metric functions and can, therefore, use an arbitrary unit of length. On the other hand, the solution~(\ref{A-f})\,--\,(\ref{V}) is invariant under the scale transformations
\begin{equation}
r\rightarrow r/\lambda,\quad m\rightarrow m/\lambda,\quad V\rightarrow \lambda^2V,\quad\;\;\lambda>0.\nonumber
\end{equation}
It means that by applying $\lambda=m$ in this transformation, we can take, as it is usually done in general relativity, the mass of a SFNS as the current unit of length. Thus, without loss of generality, we  suppose everywhere below that $m=1$.

Using~(\ref{R1})\,--\,(\ref{R3}) we find that the Kretchmann invariant diverges at the centre, namely,
\begin{equation}\label{}
K= \frac{\,5 \big(3-\sqrt{5}\,\big) \big(4a-3\sqrt{5}-3\big)^{\!2}} {6\;\!\!r^6\vphantom{\tilde{A}}}\; + \; O\big(r^{-5}\big)\,,\quad r\rightarrow0.
\nonumber
\end{equation}
In accordance with~(\ref{3m<xi0}), the condition $a>3/(\sqrt{5}-1)$ determines the subfamily of SFNSs.

The scalar field can be obtained by solving the problem $\phi'=\sqrt{\Phi'/r\,},\,\;\phi(\infty)=0$. The result is
\begin{equation}\label{phi-F}
\phi=\sqrt[4]{4/5}\, \big(F(\lambda_0,k)-F(\lambda(r),k)\big),
\end{equation}
where
\begin{equation}
k=\sqrt{1/2+1/\sqrt{5}}, \quad
\lambda_0= \arcsin\! \left(\frac{2\sqrt[4]{20}}{\,\sqrt{5}+2\,}\right),
\nonumber
\end{equation}
\begin{equation}
\lambda(r)=\arcsin\! \left(\frac{2\sqrt[4]{20}\sqrt{r(r+a)}} {\big(\sqrt{5}+2\big)r+\sqrt{5}\,a\,}\right),
\nonumber
\end{equation}
and $F(\lambda,k)$ denotes the incomplete elliptic integral of the first kind in the Legendre normal form,
\begin{equation}
F(\lambda,k)=\int_{0}^{\,\sin\lambda} \frac{dt}{\,\sqrt{1-t^2}\sqrt{1-k^2t^2}\,}.
\nonumber
\end{equation}
One has $\phi=\sqrt[4]{4/5}\,\,F(\lambda_0,k)\,-\, \sqrt{16r/5a}\,+\,O(r^{3/2})$ near $r=0$ ($\phi(0)\approx2.498$)\, and\, $\phi=2a/r+O(r^{-2})$\, at infinity. As a function of $r$, the potential can be obtained directly from equation~(\ref{V}) but has a cumbersome form; (\ref{V}) and~(\ref{phi-F}) determine, parametrically, $V(\phi)$. This potential, which has arisen in a simple demo example, is not physically interesting because it is negative everywhere in the interval $(0,\infty)$. However, there exist solutions whose self-interaction potential is positive in the region $r>r_p$ with $r_p$ sufficiently close to the centre~\cite{Nikonov2008}.

\begin{figure*}[p]
\begin{center}
\includegraphics [width=0.313\textwidth]{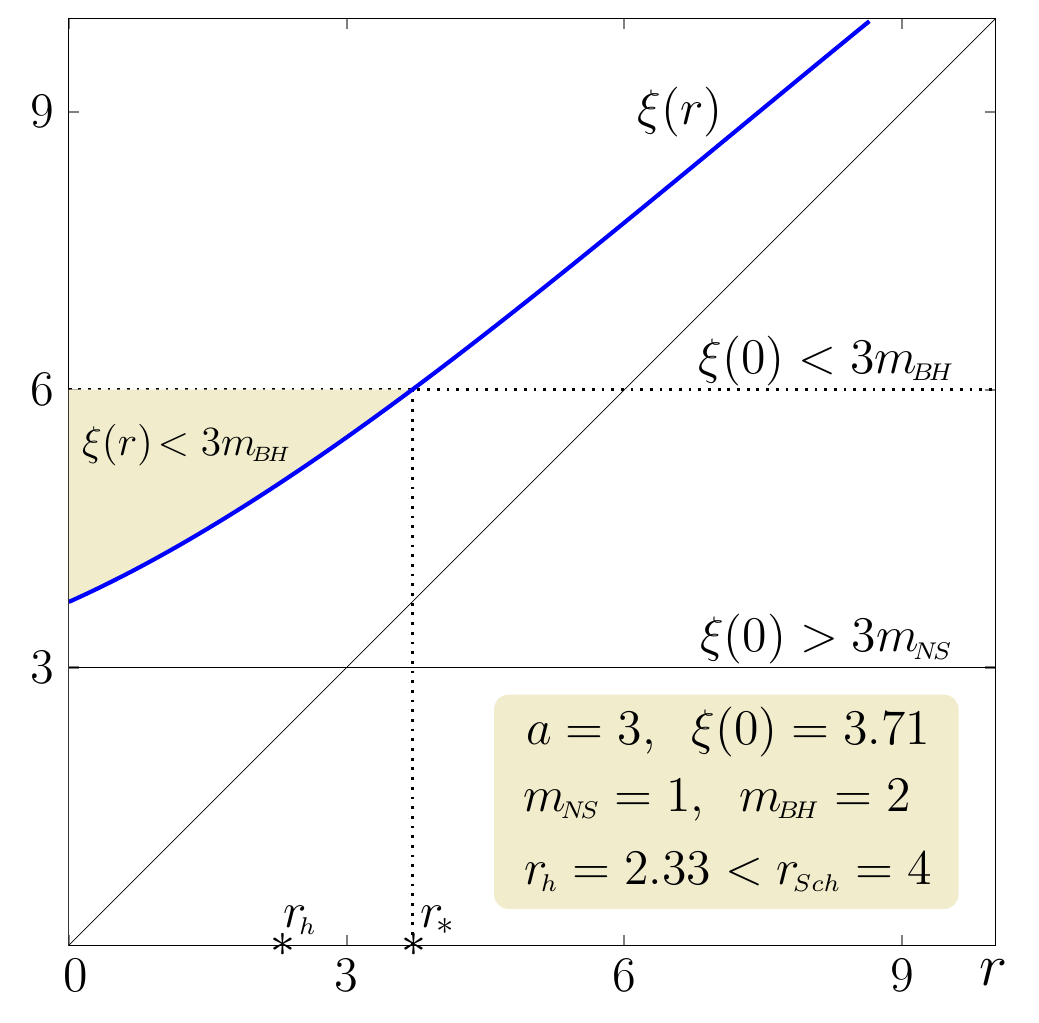}\;\;\;
\includegraphics [width=0.32\textwidth]{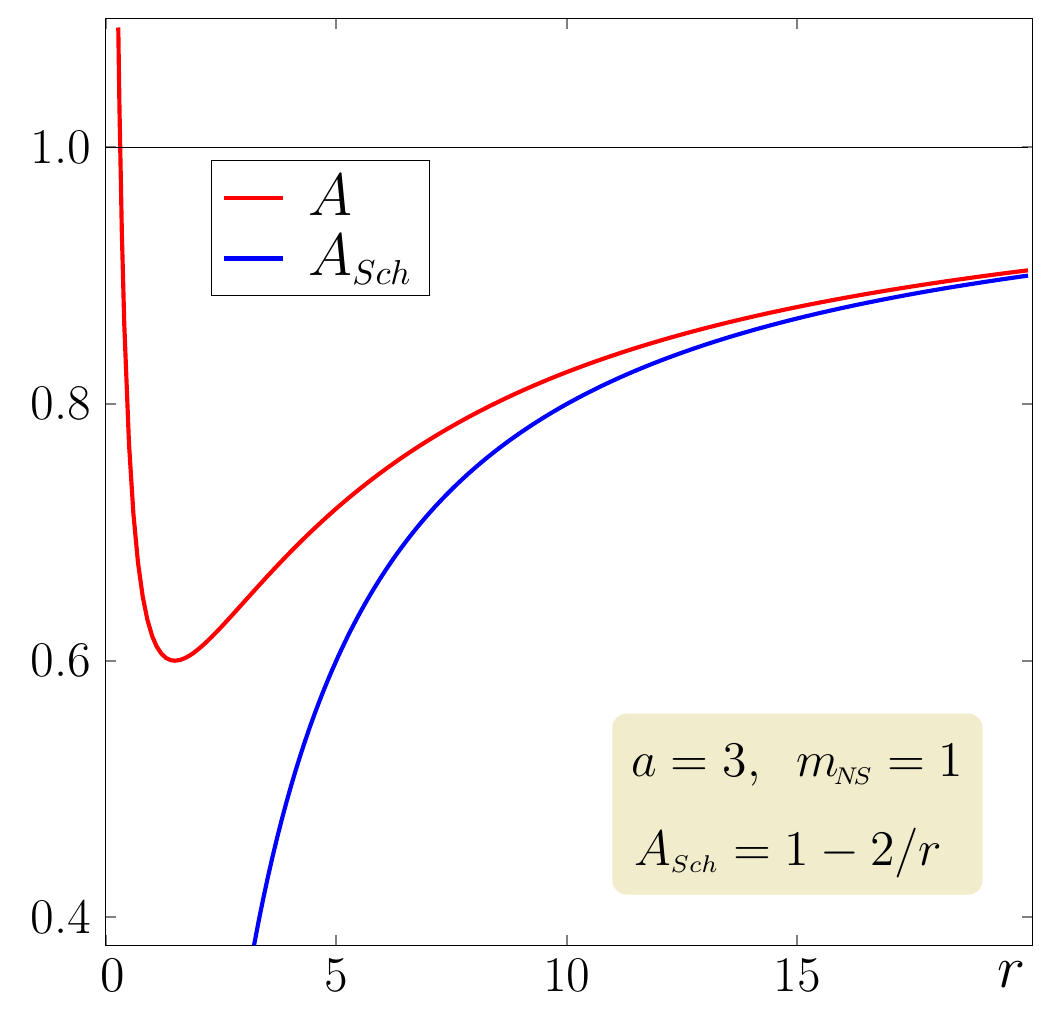}\;\;\;
\includegraphics
[width=0.319\textwidth]{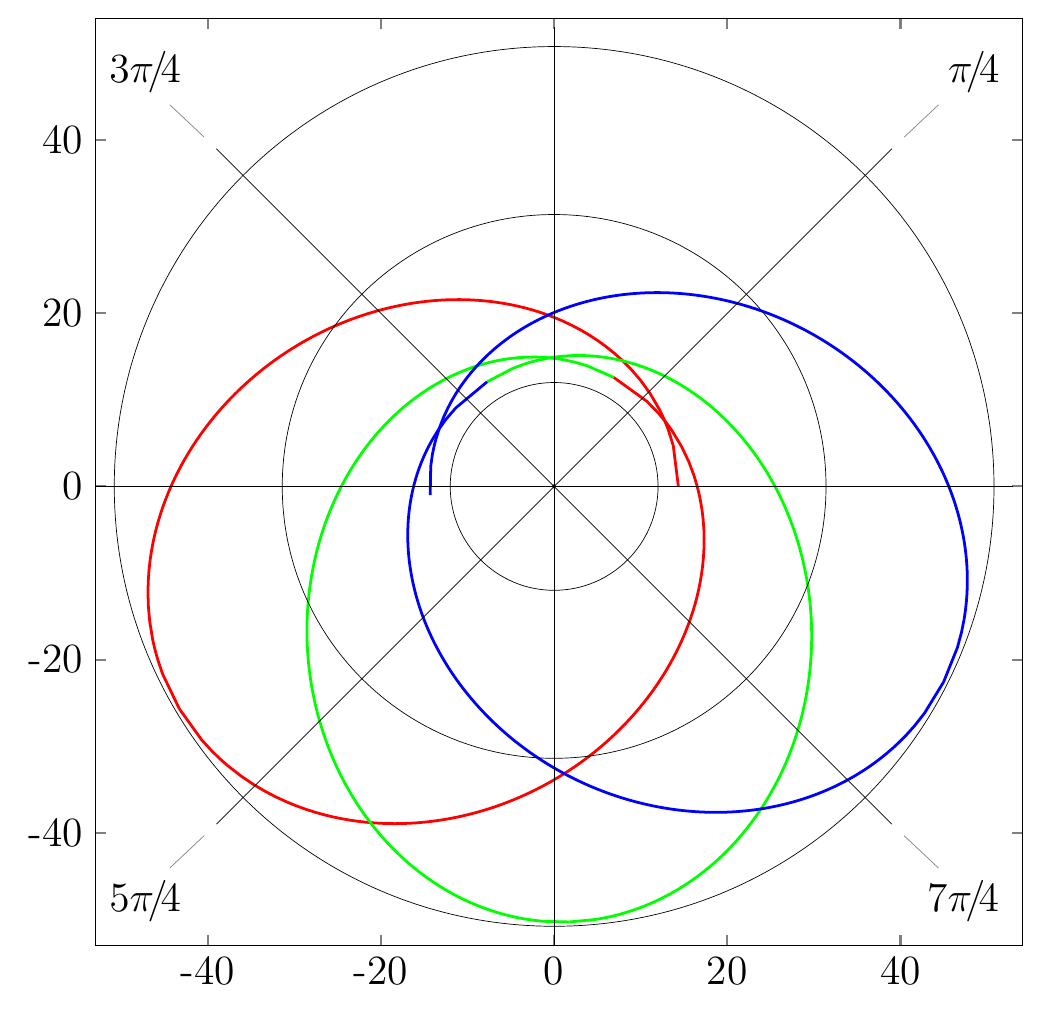}
\end{center}
\caption{The left panel shows the function $\xi(r)$ given by (\ref{xi}); SFNSs and SFBHs have masses in the intervals $(0,\xi(0))$ and $(\xi(0),\infty)$, respectively. The middle panel shows the metric functions $A(r)$ for the SFNS (\ref{Aex}) with $a=3$ and for the Schwarzschild solution of the same mass $m=1$. For this Schwarzschild spacetime, the shape of the orbit with the parameters $J=5.12$, $E=0.97$, $r_{min}=14.35$, $r_{max}=50.33$, and $\Delta\varphi=+1.05\approx\pi/3$ is plotted in the right panel.}
\label{fig1}
\end{figure*}
\begin{figure*}[p]
\begin{center}
\includegraphics [width=0.318\textwidth]{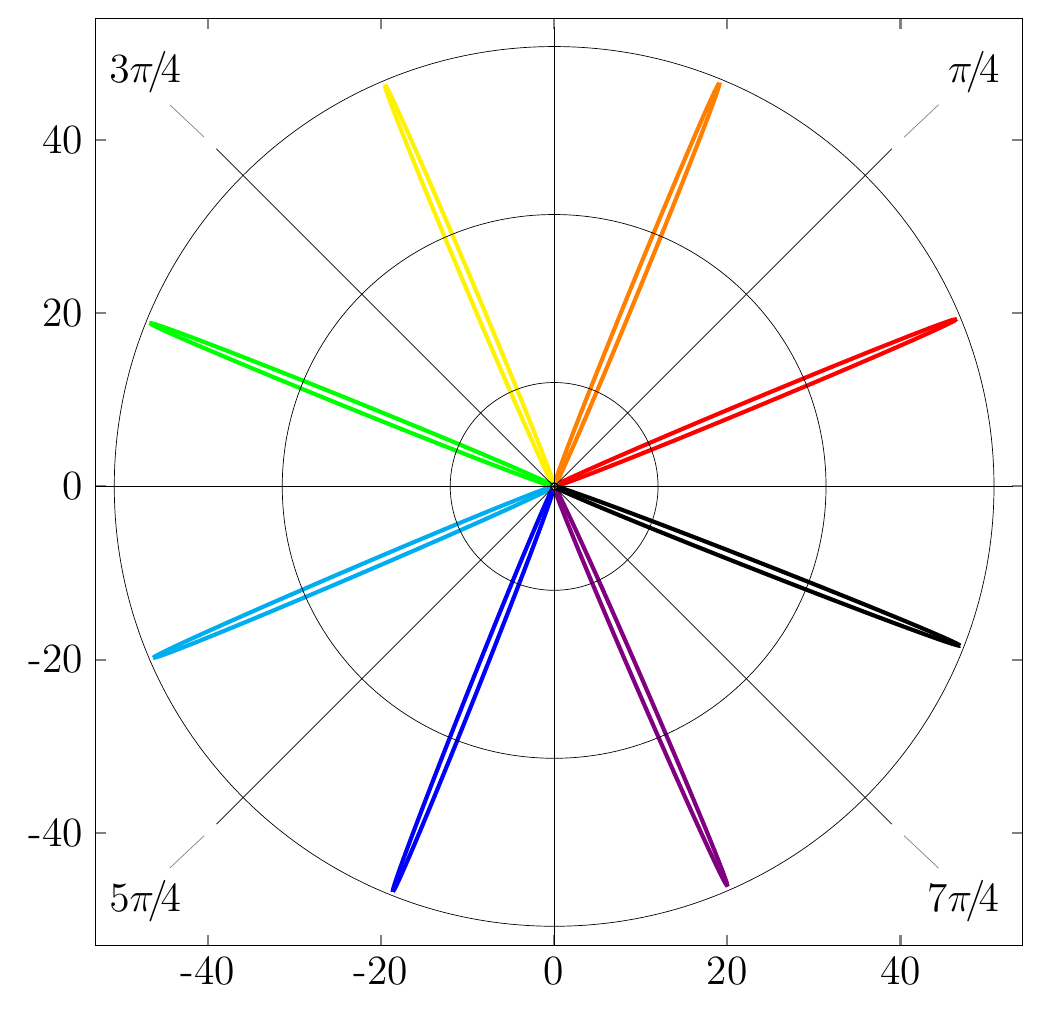}\;\;\;
\includegraphics [width=0.318\textwidth]{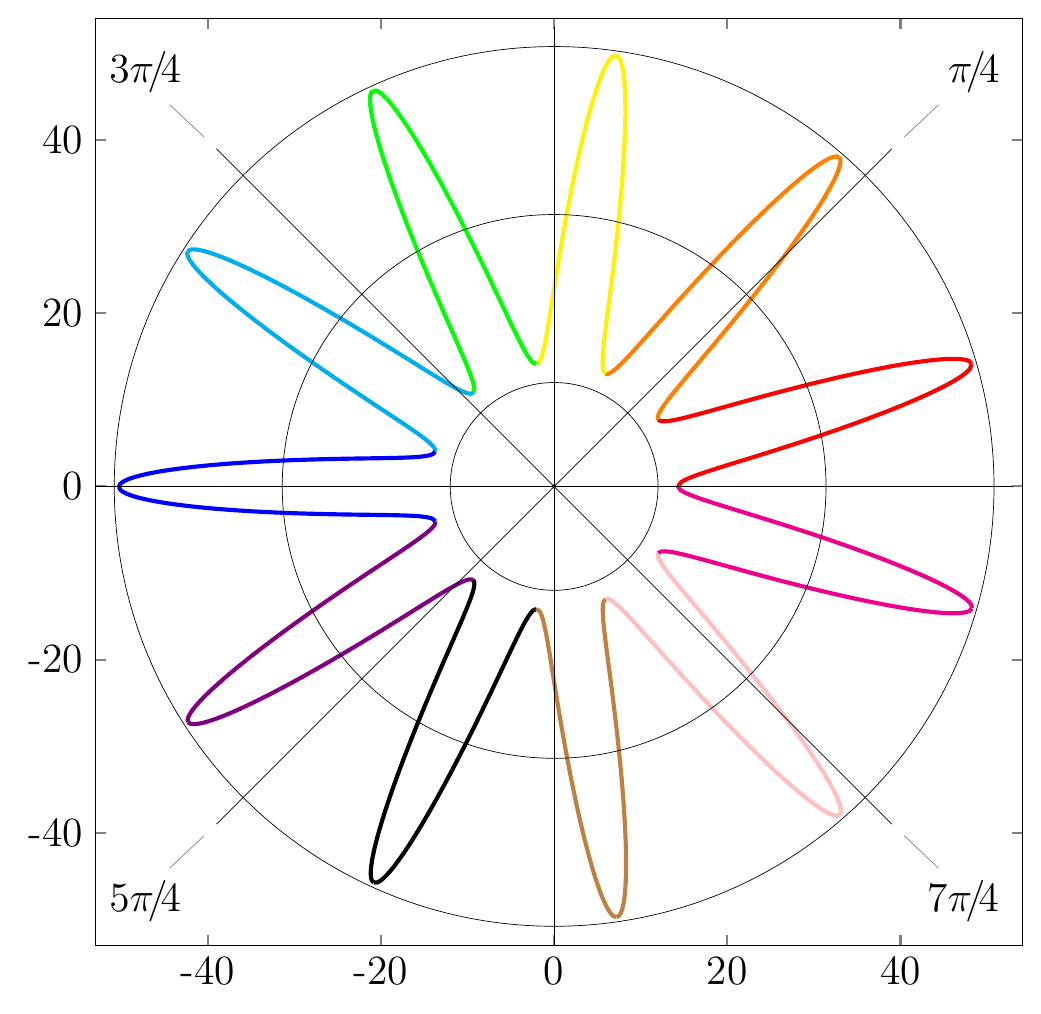}\;\;\;
\includegraphics [width=0.318\textwidth]{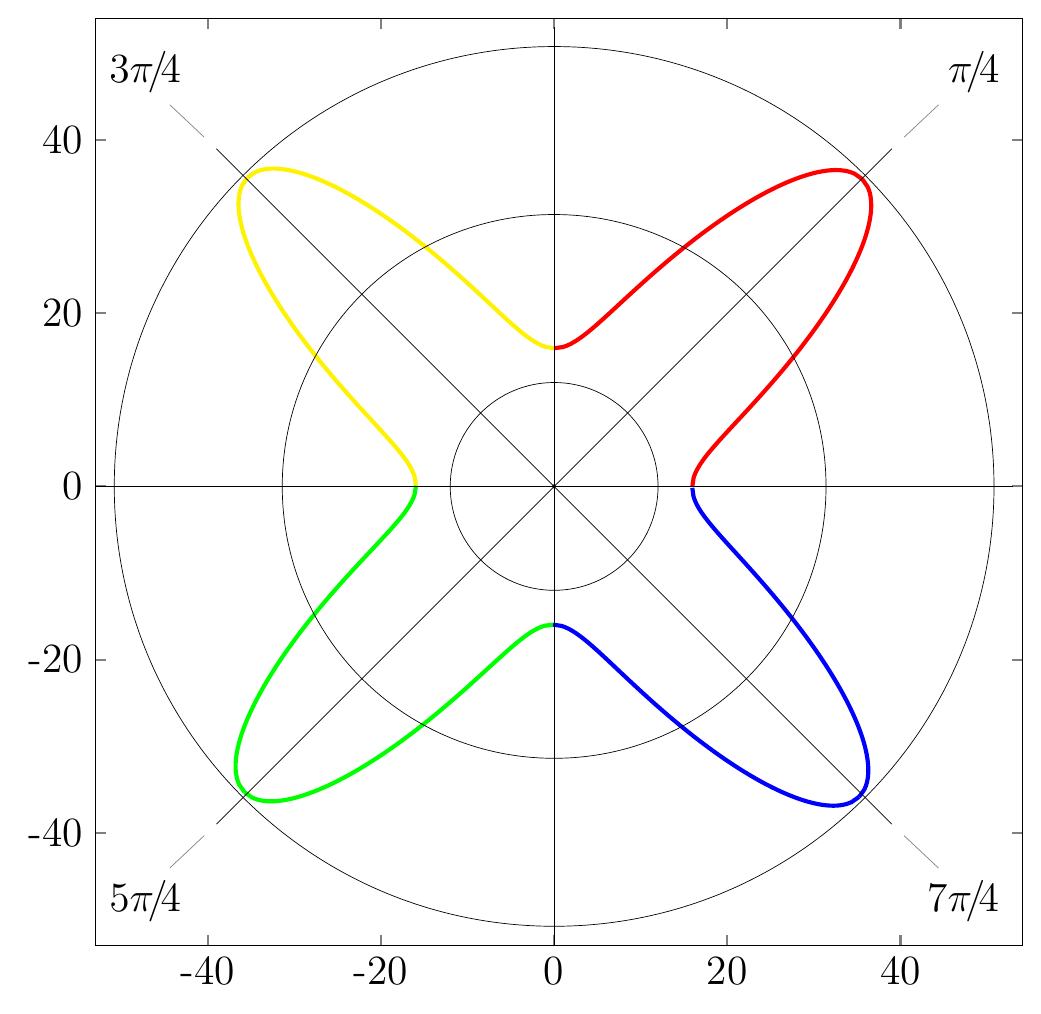}
\end{center}
\caption{The shape of orbits with different parameters and different numbers of oscillations per revolution. They all have the apocentre radius which is approximately the same as that for the Schwarzschild spacetime in Fig.~\ref{fig1}, $r_{max}\approx50.3$, but they all have a negative angle of precession.
Left panel:  $a=3$, $J=0.094$, $E=0.9606$, $r_{min}=0.37$, $r_{max}=50.37$, $\Delta\varphi=-5.495\approx-7\pi/4$.
Middle panel: $a=10$, $J=0.5$, $E=0.966$, $r_{min}=14.31$, $r_{max}=50.26$, $\Delta\varphi=-5.71\approx-20\pi/11$.
Right panel: $a=10$, $J=1.445$, $E=0.9667$, $r_{min}=15.95$, $r_{max}=50.27$, $\Delta\varphi=-4.71\approx-3\pi/2$.} \label{fig2}
\end{figure*}
\begin{figure*}[p]
\begin{center}
\includegraphics [width=0.308\textwidth]{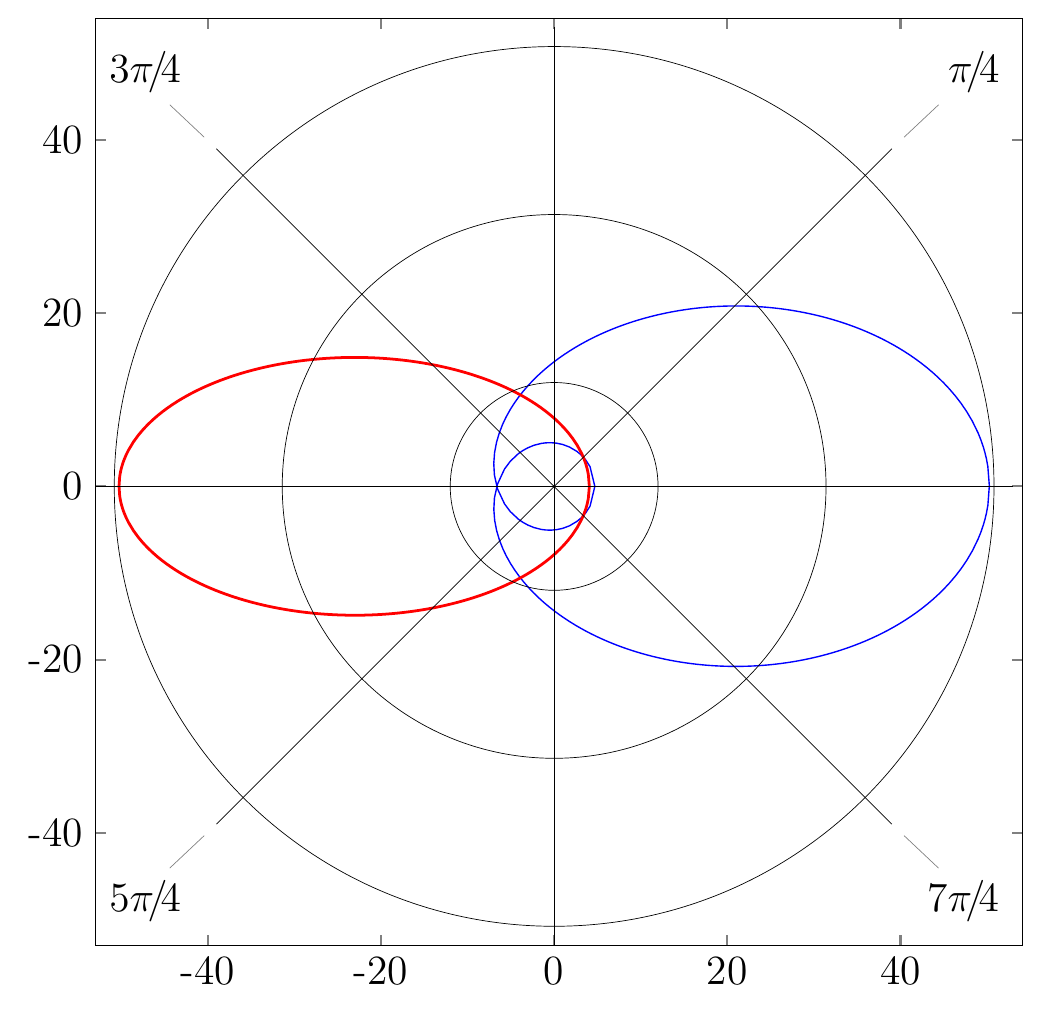}\;\;\;
\includegraphics [width=0.326\textwidth]{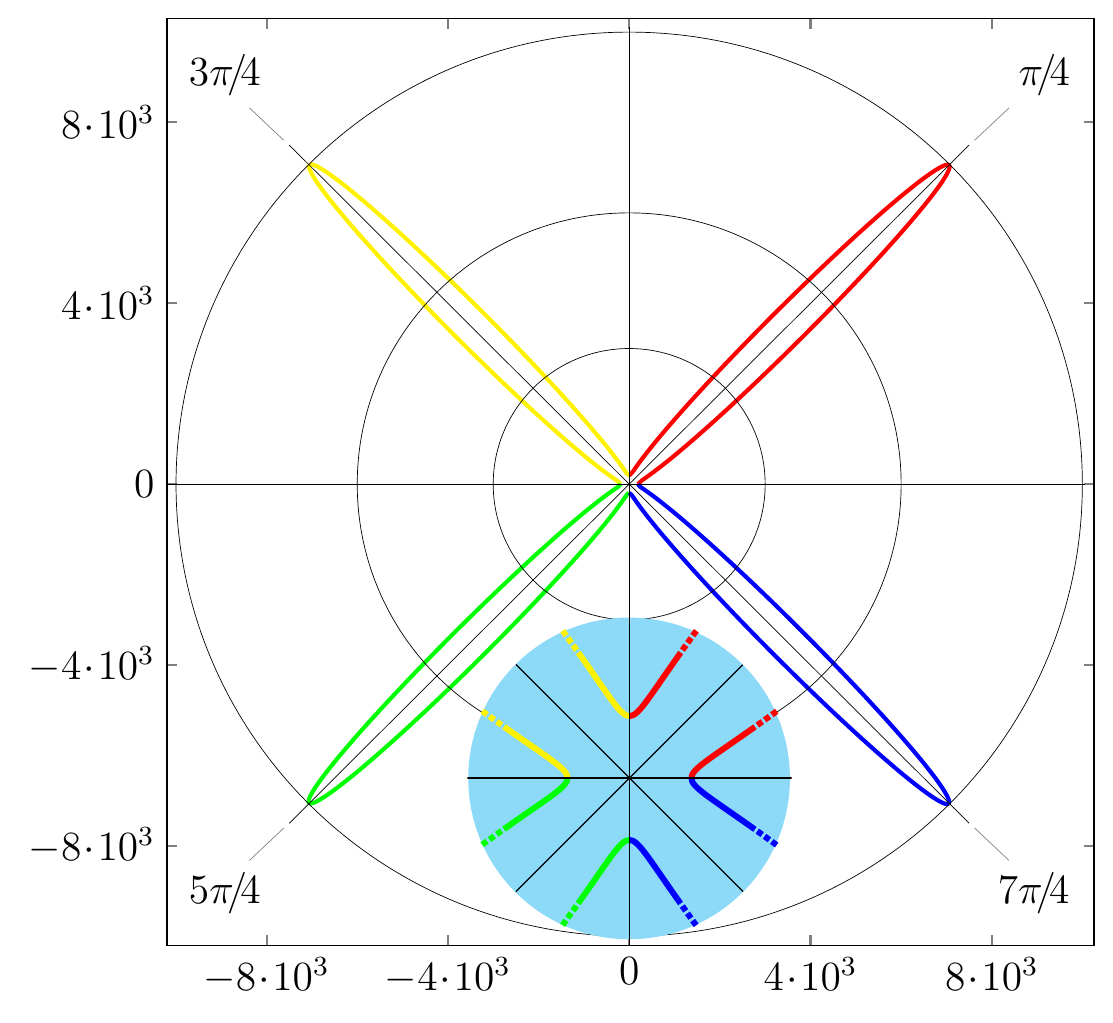}\;\;\;
\includegraphics [width=0.326\textwidth]{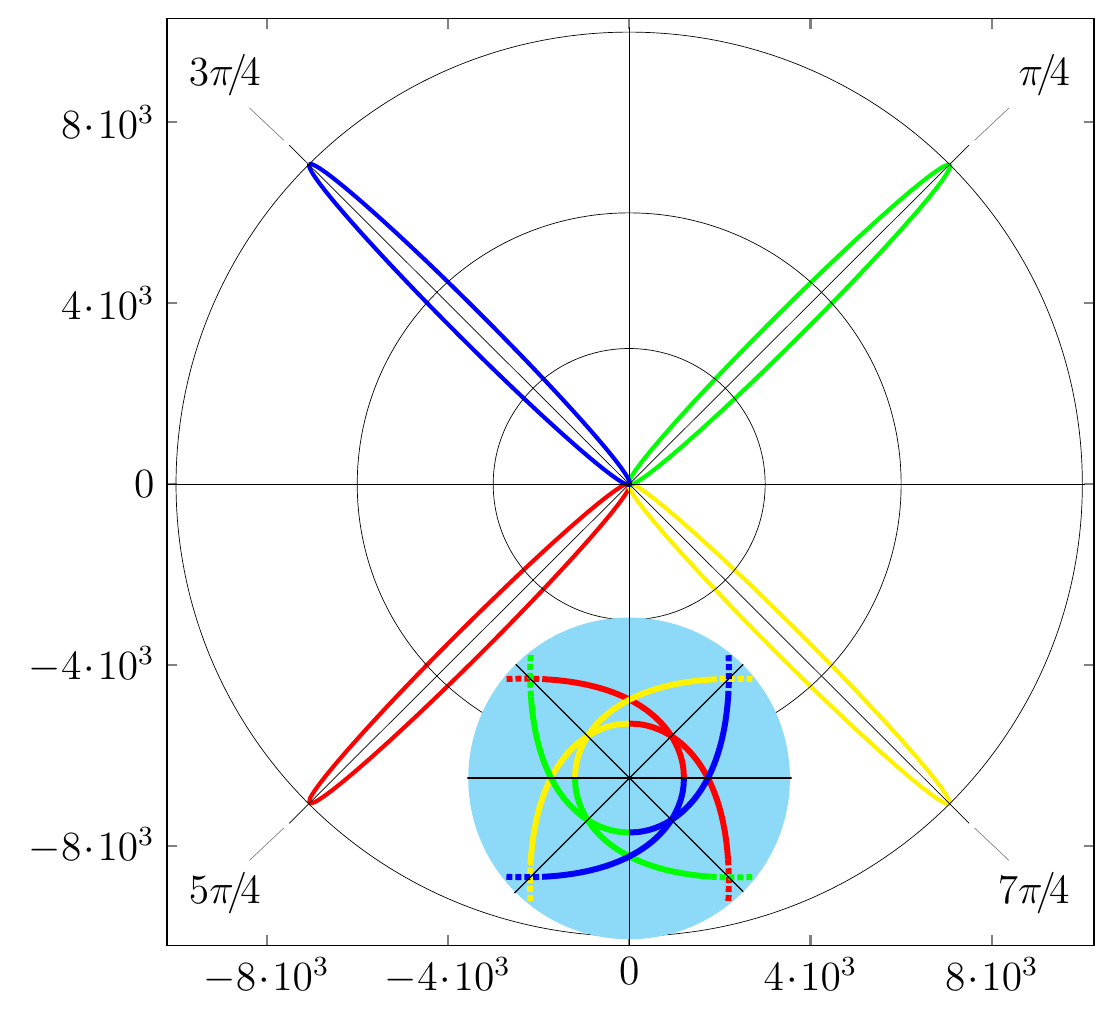}
\end{center}
\caption{Left panel: the shape of an orbit around a SFNS (left of center) and the shape of an orbit in the Schwarzschild spacetime with the same mass $m_{\scriptscriptstyle Sch}=m_{\scriptscriptstyle NS}=1$ (right of center); these orbits have the same pericentre and apocentre radii. The parameters of the orbits, respectively, are: $a=2.5$, $J=2.856$, $E=0.9635$, $r_{min}=4.07$, $r_{max}=50.25$, $\Delta\varphi=0$ (SFNS), and $J=3.88$, $E=0.966$, $r_{min}=4.70$, $r_{max}=50.25$, $\Delta\varphi=+2\pi$ (the Schwarzschild spacetime).
Middle panel: an orbit around a SFNS with parameters  $a=50$, $J=5.66$, $E=0.9998$, $r_{min}=191.1$, $r_{max}=9978.7$, $\Delta\varphi=-4.71=-3\pi/2$. Right panel: an orbit around the Schwarzschild black hole ($m_{\scriptscriptstyle Sch}=1$) with parameters $J=4.72$, $E=0.9998$, $r_{min}=8.54$, $r_{max}=9988.4$, $\Delta\varphi=+\pi/2$.} \label{fig3}
\end{figure*}

Our main goal (of considerable observational interest) is to compare orbits which have the same pericentre and apocentre radii and a comparable number of oscillations, but which either are in different spacetimes or have different parameters. Fig.~\ref{fig1} shows the typical behaviour of the function $\xi(r)$ belonging to the family~(\ref{xi}), the metric function $A(r)$ of a SFNS together (and in comparison) with the corresponding metric function of the Schwarzschild solution of the same mass, and a typical bound orbit in the latter spacetime. It is important to note that all noncircular bound orbits in the  Schwarzschild spacetime have 'the relativistic pericentre advance', that is $\Delta\varphi>0$. In contrast, in Fig.~\ref{fig2} we plot the shape of orbits possessing the same apocentre radius as the orbit in Fig.~\ref{fig1}; one of them (in the middle panel) has, in addition, the same pericentre radius. In all three cases, the angles of precession are negative and sufficiently large in magnitude. The shape of an orbit depends on the specific angular momentum $J$ and the specific energy $E$ of a test particle. Numerical simulations show that the number of oscillations per revolution decreases with increasing $J$ when the value of $E$ is fixed, as well as with decreasing $E$ when the value of $J$ is fixed; both these dependences appear to be true in general, not only for the family~(\ref{xi}).

Fig.~\ref{fig3} represents comparable orbits around SFNSs and around the corresponding Schwarzschild black hole with the same mass $m_{\scriptscriptstyle Sch}=m_{\scriptscriptstyle NS}=1$. For the SFNS with $a=2.5$, an orbit which looks like a Keplerian orbit, whose pericentre and apocentre are not shifting during one revolution of a test particle and thus $\Delta\varphi=0$, is represented in the left panel. In the Schwarzschild spacetime, the only similar orbit, having the same motionless pericentre and apocentre as the orbit around the SFNS, has the angle of precession $\Delta\varphi=2\pi$. The middle and right panels of Fig.~\ref{fig3} represent two very highly elongated orbits around the SFNS spacetime with $a=50$ and the Schwarzschild black hole, respectively. These orbits have approximately the same apocentre radius, which is larger than the size of the central part of the orbits by a factor of about 50. From the point of view of a distant observer, the orbits can look very similar or even being observationally indistinguishable from one another, whereas their angles of precession are, respectively, $-3\pi/2$ and $+\pi/2$. The insets in the middle and right panels show the obvious difference in the behaviours of the orbits in the central regions; this example explains why we need the observations of both pericentre and apocentre of elongated orbits.

\section{Conclusions}

In this paper, we considered a model of a spherically symmetric strongly gravitating massive object surrounded by a self-gravitating nonlinear scalar field, having in mind the centre of a galaxy surrounded by dark matter. This idealized model is treated in a fully analytical manner and, in this way, we have found some new features of the orbital motion of free test particles around SFNSs and SFBHs. First, for a given \textit{positive} mass, there exists a continuum of SFNSs with asymptotic geometry of the Schwarzschild spacetime with the same mass; thus, in contrast to Schwarzschild naked singularities, SFNSs exhibit the attractive nature of gravitation.

Second, the radii of the event horizon and of the innermost stable circular orbit of a SFBH are normally much less than those of the Schwarzschild black hole with the same mass. A SFNS does not have an innermost stable circular orbit but has a stable, degenerated, static orbit on which a test particle, having zero angular momentum and the minimum of its energy, remains at rest as time passes. It is important that this phenomena cannot take place on the outside of the event horizon of a vacuum black hole or a SFBH. Moreover, in the theory of self-gravitating scalar fields with the positive kinetic term in the Lagrangian, among all configurations possessing a positive mass, naked singularities and only naked singularities have such an orbit. If a noninteracting test particle is initially has a sufficiently small specific angular momentum, it will permanently remain close to the static degenerate orbit. Consequently, in a SFNS spacetime, 'slow' particles of matter, which rest on the static degenerate orbit or slowly move near it, will form a gravitationally bound cluster. One can expect that most of the energy of particles, falling into this shell with initially high relative velocities, will be radiated to infinity. If the accretion flow onto the inner region is negligible or at least sufficiently small, as in the central region of Sgr A*, then only collisions and radiative processes determine the time evolution of the cluster. Eventually the cluster will cool down and then become a spherical shell consisting of cold gas, dust, or fluid. Therefore, this shell will be seen by a distant observer as a shadow or a gray spot, which can be mistakenly taken for the shadow of a black hole.

And third, we have studied the shape of orbits close to the centres of SFNSs and shown that their angles of precession are negative, that is, pericentre retreats during each orbital revolution, and not advances as in the case of Schwarzschild black holes. At the present time, we can observe only the orbits of S-stars in the distant region of the Galactic center. However, one can hope that the future development of precise astronomical instruments (e.g., at the facilities of the Event Horizon Telescope) will accurately measure the orbital precession of the pericentres of the known S-stars and of other, at present unobserved, more short-period stars. Thus, we will be able to recognize observationally the nature of the object Sgr A*.

\end{document}